\documentclass{article}
\usepackage{amsmath,amssymb,amsthm}
\usepackage[latin1]{inputenc}
\usepackage{times}
\usepackage{mathrsfs}
\usepackage{todonotes}
\usepackage{a4wide}
\usepackage{times}
\usepackage{authblk}

\newtheorem{notations}{Notations}

\newtheorem{lem}{Lemma}

\newtheorem{theo}{Theorem}
\newtheorem{coro}{Corollary}
\newtheorem{prop}{Proposition}

\DeclareMathOperator{\K}{\mathbb{K}}
\DeclareMathOperator{\Q}{\mathbb{Q}}

\DeclareMathOperator{\C}{\mathbb{C}}
\DeclareMathOperator{\sfI}{\mathbf{I}}
\DeclareMathOperator{\jac}{jac}
\DeclareMathOperator{\F}{\mathbf{F}}
\DeclareMathOperator{\Fgenh}{\mathfrak F^{h}}
\DeclareMathOperator{\Fgena}{\mathfrak F}

\DeclareMathOperator{\frakg}{\mathfrak{g}}
\DeclareMathOperator{\frakI}{\mathfrak{I}}
\DeclareMathOperator{\frakf}{\mathfrak{f}}
\DeclareMathOperator{\homg}{\mathsf{hom}}
\DeclareMathOperator{\HS}{\mathsf{HS}}
\DeclareMathOperator{\wHS}{\mathsf{wHS}}
\DeclareMathOperator{\LM}{\mathsf{LM}}
\DeclareMathOperator{\MaxM}{\mathsf{MaxMinors}}
\DeclareMathOperator{\mon}{\mathsf{Monomials}}
\DeclareMathOperator{\dreg}{d_{reg}}
\DeclareMathOperator{\wdeg}{wdeg}

\DeclareMathOperator{\crit}{\mathsf{crit}}
\DeclareMathOperator{\DEG}{\mathsf{DEG}}
\begin{document}
\title{Critical Points and Gr\"obner Bases:\\ the Unmixed Case}
\author{
Jean-Charles Faug\`ere, Mohab Safey El Din, Pierre-Jean Spaenlehauer}
\affil{INRIA, Paris-Rocquencourt Center, PolSys Project\\
       UPMC, Univ Paris 06, LIP6\\
       CNRS, UMR 7606, LIP6\\
       UFR Ing\'enierie 919, LIP6 Passy-Kennedy\\
       Case 169, 4, Place Jussieu, F-75252 Paris\\
}
\date{}
\maketitle
\begin{abstract}
  We consider the problem of computing critical points of the
  restriction of a polynomial map to an algebraic variety. This is of
  first importance since the global minimum of such a map is reached
  at a critical point.  Thus, these points appear naturally in
  non-convex polynomial optimization which occurs in a wide range of
  scientific applications (control theory, chemistry, economics,...).

  Critical points also play a central role in recent algorithms of
  effective real algebraic geometry.  Experimentally, it has been
  observed that Gr\"obner basis algorithms are efficient to compute
  such points. Therefore, recent software based on the so-called
  Critical Point Method are built on Gr\"obner bases engines.

  Let $f_1, \ldots, f_p$ be polynomials in $ \Q[x_1, \ldots, x_n]$ of
  degree $D$, $V\subset\C^n$ be their complex variety and $\pi_1$ be
  the projection map $(x_1,\ldots, x_n)\mapsto x_1$. The critical
  points of the restriction of $\pi_1$ to $V$ are defined by the
  vanishing of $f_1, \ldots, f_p$ and some maximal minors of the
  Jacobian matrix associated to $f_1, \ldots, f_p$. Such a system is
  algebraically structured: the ideal it generates is the sum of a
  determinantal ideal and the ideal generated by $f_1,\ldots, f_p$.

  We provide the first complexity estimates on the computation of
  Gr\"obner bases of such systems defining critical points. We prove
  that under genericity assumptions on $f_1,\ldots, f_p$, the complexity is
  polynomial in the generic number of critical points, i.e.
  $D^p(D-1)^{n-p}{{n-1}\choose{p-1}}$. More particularly, in the
  quadratic case $D=2$, the complexity of such a Gr\"obner basis
  computation is polynomial in the number of variables $n$ and
  exponential in $p$.  We also give experimental evidence supporting
  these theoretical results.
\end{abstract}

\section{Introduction}

\textbf{Motivations and problem statement.} 
The local extrema of the restriction of a polynomial map to a real algebraic variety are reached at the critical points of the map under consideration. Hence, computing these critical points is of first importance for polynomial optimization which arises in a wide range of applications in engineering sciences (control theory, chemistry, economics, etc.).

Computing critical points is also the cornerstone of algorithms for asymptotically optimal algorithms for polynomial system solving over the reals (singly exponential in the number of variables).  Indeed, for computing sample points in each connected component of a semi-algebraic set, the algorithms based on the so-called critical point method rely on a reduction of the initial problem to polynomial optimization problems. In \cite{BPR96, BPR98} (see also \cite{GV88, HRS89, HRS93}), the best complexity bounds are obtained using infinitesimal deformation techniques of semi-algebraic geometry, nevertheless obtaining efficient implementations of these algorithms remains an issue.

Tremendeous efforts have been made to obtain fast implementations relying on the critical point method (see \cite{SaSc03, EveLazLazSaf09, S07, HS2011, HonSaf09, FauMorRouSaf08, SaSc04}). This is achieved with techniques based on algebraic elimination and complex algebraic geometry.  For instance, when the input polynomial system $(\mathbf{F}): f_1=\cdots=f_p=0$ in $\Q[x_1, \ldots, x_n]$ satisfies genericity assumptions, one is led to compute the set of critical points of the restriction of the projection $\pi_1:(x_1, \ldots, x_n)\rightarrow x_1$ to the algebraic variety $V(\F)\subset \C^n$ defined by $\mathbf{F}$; this set is denoted by $\crit(\pi_1, V(\F))$.

The set $\crit(\pi_1, V(\F))$ is defined by $\mathbf{F}$ and the vanishing of the maximal minors of the truncated Jacobian matrix of $\mathbf{F}$ obtained by removing the partial derivatives with respect to $x_1$. This system is highly-structured: algebraically, we are considering the sum of a determinantal ideal with the ideal $\langle f_1, \ldots, f_p\rangle$. 

In practice, we compute a rational parametrization of this set through Gr\"obner bases computations which are fast in practice. 
We have observed that the behavior of Gr\"obner bases on these systems does not coincide with the generic one. In the particular case of quadratic equations, it seems to be polynomial in $n$ and exponential in $p$ which meets the best complexity known bound for the quadratic minimization problem \cite{barvinok1993, GriPas05}.
Understanding the complexity of these computations is a first step towards the design of dedicated Gr\"obner bases algorithms, so we focus on the following important open problems:
\begin{itemize}
\item[{\bf (A)}] Can we provide {\em complexity estimates} for the computation of Gr\"obner bases of ideals defined by such {\em structured algebraic systems}?
\item[{\bf (B)}] Is this computation {\em polynomial in the generic number of cri\-tical points}? 
\item[{\bf (C)}] In the {\em quadratic case}, is this computation {\em polynomial in the number of variables} (and exponential in the codimension)? 
\end{itemize}
Under genericity assumptions, we actually provide affirmative answers to all these questions.

\smallskip

\textbf{Computational methodology and related complexity issues.}
Gr\"obner bases are computed using multi-modular arithmetics and we will focus only on arithmetic complexity results; so we may consider systems defining critical points with coefficients not only in $\Q$ but also in a prime field. 

Let $\mathbb{K}$ be a field, $\overline{\mathbb{K}}$ be its algebraic closure and $\mathbf{F}=(f_1, \ldots, f_p)$ be a family of polynomials in $\mathbb{K}[x_1, \ldots, x_n]$ of degree $D$ and $V(\F)$ be their set of common zeroes in $\overline{\mathbb{K}}^n$. 

We denote the Jacobian matrix
$$
\left [\begin{array}{ccc}
  \frac{\partial f_1}{\partial x_1} & \cdots &\frac{\partial f_1}{\partial x_n} \\
\vdots & &\vdots \\
  \frac{\partial f_p}{\partial x_1} & \cdots &\frac{\partial f_p}{\partial x_n} \\
\end{array}\right ]
$$
by $\jac(\F)$ and the submatrix obtained by removing the first $i$ co\-lumns by $\jac(\F, i)$. The set of maximal minors of a given rectangular matrix ${\sf M}$ will be denoted by $\MaxM({\sf M})$. 

Finally, let $\sfI(\F, 1)$ be the ideal $\langle \F\rangle+\langle \MaxM(\jac(\F, 1))\rangle$. When $\F$ is a reduced regular sequence and $V(\F)$ is smooth, the algebraic variety associated to $\sfI(\F, 1)$ is exactly $\crit(\pi_1, V(\F))$. 

So, to compute a rational parametrization of $\crit(\pi_1, V(\F))$, we use the classical solving strategy which proceeds in two steps: 
\begin{itemize}
  \item[{\em (i)}] compute a Gr\"obner basis for a {\em grevlex} ordering of $\sfI(\F, 1)$ using the ${F}_5$ algorithm (see \cite{Fau02});
  \item[{\em (ii)}] use the FGLM algorithm \cite{FauGiaLazMor93, FauMou11} to obtain a Gr\"obner basis of $\sfI(\F, 1)$ for the lexicographical ordering or a rational paramet\-rization of $\sqrt{\sfI(\F, 1)}$.  
\end{itemize}
Algorithm $F_5$ (Step {\em (i)}) computes Gr\"obner bases by
row-echelon form reductions of submatrices of the Macaulay matrix up
to a given degree.  This latter degree is
called {\em degree of regularity}. When the input satisfies regularity
properties, this complexity of this step can be analyzed by estimating the degree of regularity.

FGLM algorithm \cite{FauGiaLazMor93} (Step {\em (ii)}) and its recent efficient variant \cite{FauMou11} are based on computations of characteristic polynomials of linear endomorphisms in $\mathbb{K}[x_1, \ldots, x_n]/\sfI(\F, 1)$. This is done by performing linear algebra operations of size the {\em degree of $\sfI(\F, 1)$} (which is the number of solutions counted with multiplicities). 

Thus, we are faced to the following problems:
\begin{itemize}
\item[{\em (1)}] estimate the degree of regularity of the ideal generated by the homogeneous components of highest degree of the set of generators $\F, \MaxM(\jac(\F, 1))$; 
\item[{\em (2)}] show that the above estimation allows to bound the complexity of computing a {\em grevlex} Gr\"obner basis of $\sfI(\F, 1)$;  
\item[{\em (3)}] provide sharp bounds on the degree of the ideal $\sfI(\F, 1)$. 
\end{itemize}
As far as we know, no results are known for problems {\em (1)} and {\em (2)}. 
Problem {\em (3)} has already been investigated in the literature: see \cite{NieRan09} where some bounds are given on the cardinality of $\crit(\pi_1, V(\F))$. We give here a new algebraic proof of these bounds. 

\smallskip
  
\textbf{Main results.} 
Let $\K[x_1, \ldots, x_n]_D$ denote $\{f\in \K[x_1, \ldots, x_n]\mid \deg(f)=D\}$ and remark that it is a finite-dimensional vector space.
In the following, we solve the three aforementioned problems under a {\em genericity} assumption on $\F$: we actually prove that there exists a non-empty Zariski open set $\mathscr{O}\subset\overline{\K}[x_1, \ldots, x_n]_D^p$ such that for all $\F\in \mathscr{O}$: 
\begin{itemize}
\item[{\em (1)}] the degree of regularity of the ideal generated by the homogeneous components of largest degree of $\F, \MaxM(\jac(\F, 1))$ is $\dreg=D(p-1)+(D-2)n+2$ (see Theorem \ref{theo:hdreg});
\item[{\em (2)}] with the $F_5$ algorithm, the highest degree reached during the computation is bounded by $\dreg$ (see Theorem \ref{theo:dregaff}); 
\item[{\em (3)}] the degree of $\sfI(\F, 1)$ is $\leq \delta=D^p(D-1)^{n-p}{{n-1}\choose{p-1}}$. 
\end{itemize}
The degree of regularity given in {\em (1)} is obtained thanks to an explicit formula for the Hilbert series of the homogeneous ideal under consideration (see Proposition \ref{prop:HSunmixed}). This is obtained by taking into account the determinantal structure of some of the generators of the ideal we consider.  
The above estimates are the key results which enable us to provide positive answers to questions {\bf A}, {\bf B} and {\bf C} under genericity assumptions. 

Before stating complexity results on the computation of critical points with Gr\"obner bases, we need to introduce a standard notation. 
Let $\omega$ be a real number such that a row echelon form of a $n\times n$-matrix with entries in $\K$ is computed within $O(n^\omega)$ arithmetic operations in $\K$.

 We prove that there exists a non-empty Zariski open set $\mathscr{O}\subset \overline{\K}[x_1, \ldots, x_n]_D^p$ such that for all $\F\in \mathscr{O}\cap \K[x_1, \ldots, x_n]^p$:
\begin{itemize}
  \item[{\bf (A)}] computing a {\em grevlex} Gr\"obner basis of $\sfI(\F, 1)$ can be done within $O\left ({{n+\dreg}\choose{n}}^\omega\right )$ arithmetic operations in $\K$ (see Theorem \ref{theo:compl}); 
  \item[{\bf (B)}] computing a rational parametrization of $\crit(\pi_1, V(\F))$ using Gr\"obner bases can be done within $O\left (\delta^{4.03\omega}\right )$ arithme\-tic operations in $\K$ (see Corollary \ref{coro:uniformbound});
  \item[{\bf (C)}] when $D=2$ (quadratic case), a rational paramet\-rization of $\crit(\pi_1, V(\F))$ using Gr\"obner bases can be computed with\-in $O\left( \binom{n+2p}{2p}^\omega+n 2^{3p}\binom{n-1}{p-1}^3\right)$ arithmetic operations in $\K$, this is polynomial in $n$ and exponential in $p$ (see Corollary~\ref{coro:compln}). 
\end{itemize}
We also provide more accurate complexity results. The uniform complexity bound given for answering question {\bf (B)} is rather pessimistic. The exponent $4.03\omega$ being obtained after majorations which are not sharp; numerical experiments are given to support this (see Section \ref{sec:experiments}). Moreover, under the above genericity assumption, we prove that, when $p$ and $D$ are fixed, computing a rational parametrization of $\crit(\pi_1, V(\F))$ using Gr\"obner bases is done within $O(D^{3.57n})$ arithmetic operations in $\K$ (see Corollary \ref{coro:complD}).

We also give timings for computing grevlex and lex Gr\"obner bases of $\sfI(\F,1)$ with the \textsc{Magma} computational algebra system and with the FGb library when $\mathbb K=\mathsf{GF}(65521)$. These experiments show that the theoretical bounds on the degree of regularity and on the degree of $\sfI(\F,1)$ (Theorem \ref{theo:dregaff}) are sharp. They also provide some indication on the size of problems that can be tackled in practice: e.g. when $D=2$ and $p=3$ (resp. $D=3$ and $p=1$), random dense systems with $n\leq 21$ (resp. $n\leq 14$) can be tackled (see Section \ref{sec:experiments}).

\smallskip

\textbf{Related works.} As far as we know, dedicated complexity analysis of Gr\"obner bases on ideals defining critical points has not been investigated before. However, as we already mentioned, the determinantal structure of the system defining $\crit(\pi_1, V(\F))$ plays a central role in this paper.  

In \cite{FauSafSpa10a}, we provided complexity estimates for the computation of Gr\"obner bases of ideals generated by minors of a linear matrix. This is generalized in \cite{FauSafSpa11a} for matrices with entries of degree $D$. Nevertheless, the analysis which is done here differs significantly from these previous works. Indeed, in \cite{FauSafSpa10a, FauSafSpa11a} a genericity assumption is done on the entries of the considered matrix. We cannot follow the same reasonings since $\MaxM(\jac(\F, 1))$ depends on $\F$. Nevertheless, it is worthwhile to note that, as in \cite{FauSafSpa10a, FauSafSpa11a}, we use properties of determinantal ideals given in \cite{ConHer94}.

Bounds on the number of critical points (under genericity assumptions) are given in \cite{NieRan09} using the Giambelli-Thom-Porteous degree bounds on determinantal varieties (see \cite[Ex. 14.4.14]{Ful97}). 

In \cite{barvinok1993}, the first polynomial time algorithms in $n$ for deciding emptiness of a quadratic system of equations over the reals is given.  Further complexity results in the quadratic case for effective real algebraic geometry have been given in \cite{GriPas05}. In the general case, algorithms based on the so-called critical point method are given in \cite{BPR96, BPR98, GV88, HRS89, HRS93}. Critical points defined by systems $\F, \MaxM(\jac(\F, 1))$ are computed in algorithms given in \cite{BGHM1, BGHM2, BGHM3, BGHM4, BanGiuHeiSafSch10, SaSc03, ARS, FauMorRouSaf08}. The {\sf RAGlib} maple package implements the algorithms given in \cite{SaSc03, FauMorRouSaf08} using Gr\"obner bases. 

The systems $\F, \MaxM(\jac(\F, 1))$ define polar varieties: indeed, this notion coincides with critical points in the regular case). In \cite{BGHM1, BGHM2, BGHM3, BGHM4, BanGiuHeiSafSch10}, rational parametrizations are obtained using the geometric resolution algorithm \cite{GiuLecSal01} and a local description of these polar varieties. This leads to algorithms computing critical points running in probabilistic time polynomial in $D^p(p(D-1))^{n-p}$. Note that this bound for $D=2$ and $p=n/2$ is not satisfactory. In this paper, we also provide complexity estimations for computing critical points but using Gr\"obner bases, which is the engine we use in practice. Our results provide an explanation to the good practical behavior we have observed.

\smallskip
\textbf{Organization of the paper.}  Section \ref{sec:prelim} recalls well-known properties of generic polynomial systems.
Problems {\em (1)} and {\em (2)} mentioned above are respectively tackled in Sections \ref{sec:hom} and \ref{sec:affine}. Problem {\em (3)} is solved at the end of Section \ref{sec:affine}. Complexity results are derived in Section \ref{sec:compl}. 
Experimental results supporting the theoretical results are given in Section~\ref{sec:experiments}.

\smallskip
\textbf{Conclusions and Perspectives.} We give new bounds on the degree of regularity and an explicit formula for the Hilbert series of the ideal vanishing on the critical points under genericity assumptions. This leads to new complexity bounds for computing Gr\"obner bases of these ideals.

However, we only considered the \emph{unmixed
  case}: all polynomials $f_1,\ldots, f_p$ share the same degree $D$. The \emph{mixed case} (when the degrees of the polynomials $f_1,\ldots, f_p$ are different) cannot be treated similarly since the difference of the degrees induce a combinatorial structure which has to be investigated.
We intend to investigate this question in future works using the Eagon-Northcott complex, which yields a free resolution of the ideal generated by the maximal minors of a polynomial matrix under genericity assumptions.  
{F}rom this, we also expect to obtain a variant of the $F_5$ algorithm dedicated to these ideals. 

\medskip

\textbf{Acknowledgments.} This work was supported in part by the
GeoLMI grant (ANR 2011 BS03 011 06) and by the EXACTA grant
(ANR-09-BLAN-0371-01) of the French National Research Agency.

 \section{Preliminaries}
\label{sec:prelim}

\begin{notations}\label{notations:sec:prelim}
 The set of variables $\{x_1,\ldots, x_n\}$ is
denoted by $X$.
For $d\in \mathbb
N$, $\mon(d)$ denotes the set of monomials of degree $d$ in the
polynomial ring $\mathbb K[X]$ (where $\mathbb K$ is a
field, its algebraic closure being denoted by $\overline{\mathbb K}$).
We let $\mathfrak a$ denote the finite set of parameters $\{\mathfrak
a_{\mathfrak m}^{(i)} : 1\leq i\leq p, \mathfrak m\in \bigcup_{0\leq
  d\leq D}\mon(d)\}$.

We also introduce the following generic systems:
\begin{itemize}
\item $\Fgena=(\frakf_1,\ldots, \frakf_p)\in \mathbb K(\mathfrak a)[X]^p$ is the generic polynomial system of degree $D$:
$$\mathfrak f_i=\sum_{\substack{\mathfrak m \text{ monomial}\\\deg(\mathfrak m)\leq D}}\mathfrak a_{\mathfrak m}^{(i)}\mathfrak m; $$
\item $\Fgenh=(\frakf_1^h,\ldots, \frakf_p^h)\in \mathbb K(\mathfrak a)[X]^p$ is the generic \emph{homogeneous} polynomial system of degree $D$:
$$\mathfrak f_i=\sum_{\substack{\mathfrak m \text{ monomial}\\\deg(\mathfrak m) = D}}\mathfrak a_{\mathfrak m}^{(i)}\mathfrak m.$$
\end{itemize}

We let $V(\F)\subset \overline{\mathbb K}^n$ denote the variety of
$\F=(f_1,\ldots, f_p)$. The projective variety of a homogeneous family
of polynomials $\F^h$ is denoted by $W(\F^h)$. The projection on the
first coordinate is denoted by $\pi_1$, and the critical points of the
restriction of $\pi_1$ to $V(\F)$ are denoted by
$\crit(\pi_1,V(\F))\subset V(\F)$.  Also, $\sfI(\F,1)$ denotes the
ideal generated by $\F$ and by the maximal minors of the truncated
Jacobian matrix $\jac(\F,1)$.

Throughout the paper, if $R$ is a ring and $I\subset R$ is an ideal, we call \emph{dimension} of $I$ the Krull dimension of the quotient ring $R/I$. 
\end{notations}

The goal of this section is to prove that the ideal $\sfI(\Fgenh,1)$ is $0$-dimensional. This will be done in Lemma \ref{lem:homzerodim} below; to do that we will use geometric statements of Sard's theorem which require $\K$ to have characteristic $0$. This latter assumption can be weakened using algebraic equivalents of Sard's Theorem (see \cite[Corollary 16.23]{Eis95}).

\begin{lem}\label{lem:smooth}
  Let $\sfI(\Fgena,0)$ be the ideal generated by
  $\Fgena$ and by the maximal minors of its Jacobian matrix.  Then its
  variety $V(\sfI(\Fgena,0))\subset\overline{\mathbb K(\mathfrak a)}^n$ is empty
  and hence $V(\Fgena)$ is \emph{smooth}.
\end{lem}

\begin{proof}
  To simplify notations hereafter, we denote by $h_1, \ldots, h_p$ the polynomials obtained from ${\frak f}_1, \ldots, {\frak f}_p$ by removing their respective constant terms ${\frak a}_1^{(1)}, \ldots , {\frak a}_1^{(p)}$. We will also denote by $\mathscr{A}$ the remaining parameters in $h_1, \ldots, h_p$.
Let $\psi$ denote the mapping
$$\begin{array}{r@{~}c@{~}c@{~}c}
  \psi :& \overline{\mathbb K(\mathscr{A})}^n &\longrightarrow&\overline{\mathbb K(\mathscr{A})}^p\\
  & \mathbf c&\longmapsto& (h_1(\mathbf c),\ldots,h_p(\mathbf c))
\end{array}$$

Suppose first that $\psi(\overline{\mathbb K(\mathscr{A})}^n)$ is not
dense (for the Zariski topology) in $\overline{\mathbb K(\mathscr{A})}^p$. Since the image
$\psi(\overline{\mathbb K(\mathscr{A})}^n)$ is a constructible set, it
is contained in a proper Zariski closed subset $\mathscr W\subset \overline{\mathbb
  K(\mathscr{A})}^p$. Since there is no algebraic relation between 
${\frak a}_1^{(1)}, \ldots , {\frak a}_1^{(p)}$ and the parameters in $\mathscr{A}$, this implies that the variety defined by $h_1+{\frak a}_1^{(1)}= \cdots= h_p+{\frak a}_1^{(p)}$ is empty and consequently smooth. Since $h_i+{\frak a}_i^{(1)}= {\frak f}_i$, our statement follows.

Suppose now that $\psi(\overline{\mathbb K(\mathscr{A})}^n)$ is dense in
$\overline{\mathbb K(\mathscr{A})}^p$.  Let $K_0\subset
\overline{\mathbb K(\mathscr{A})}^p$ be the set of critical values of
$\psi$.  By Sard Theorem \cite[Chap. 2, Sec. 6.2, Thm 2]{Sha88}, $K_0$
is contained in a proper closed subset of $\overline{\mathbb
  K(\mathscr{A})}^p$. Again, there is no algebraic relation between 
${\frak a}_1^{(1)}, \ldots , {\frak a}_1^{(p)}$ and the parameters in $\mathscr{A}$. 
Consequently, the variety
associated to the ideal generated by the system $\mathfrak
f_1,\ldots,\mathfrak f_p$ and by the
maximal minors of $\jac({\frak F})$ is empty. 
\end{proof}

\begin{coro}\label{coro:smoothHom}
Let $\sfI(\Fgenh,0)$
be the ideal generated by $\Fgenh$ and by the maximal minors of its Jacobian
matrix.  Then the associated projective variety $W(\sfI(\Fgenh,0))\subset\mathbb P^{n-1} \overline{\mathbb K(\mathfrak a)}$ is empty.
\end{coro}

\begin{proof}
  For $1\leq i\leq n$, we denote by $O_i$ the set  $$\{(c_1:\ldots:c_n)\mid c_i\neq 0\}\subset \mathbb P^{n-1}\overline{\mathbb
  K(\mathfrak a)}$$ and we consider the canonical open covering of
  $\mathbb P^{n-1} \overline{\mathbb K(\mathfrak a)}$:
$$\mathbb P^{n-1} \overline{\mathbb K(\mathfrak a)} = \bigcup_{1\leq i\leq n} O_i. $$
Therefore $W(\sfI(\Fgenh,0))=\bigcup_{1\leq i\leq n}(W(\sfI(\Fgenh,0))\cap O_i)$.
Denote by $\Fgena_i$ the system obtained by substituting the variable $x_i$ by $1$ in $\Fgenh$. According to 
Lemma \ref{lem:smooth} applied to $\Fgena_i$, the variety $V(\sfI(\Fgena_i,0))$ is empty. Therefore, the set $W(\sfI(\Fgenh,0))\cap O_i$ is also empty.
Consequently,
$W(\sfI(\Fgenh,0))=\emptyset$. 
\end{proof}
We can now deduce the following result. 
\begin{lem}\label{lem:homzerodim}
The projective variety $W(\sfI(\Fgenh, 1)) \subset\mathbb P^{n-1}\overline{\mathbb K(\mathfrak a)}$ 
is empty, and hence $\dim(\sfI(\Fgenh, 1))=0$.
\end{lem}

\begin{proof}
  We let $\varphi_0$ and $\varphi_1$ denote the two following morphisms:
$$\begin{array}{rrcl}
\varphi_0:&\mathbb K(\mathfrak a)[x_1,\ldots, x_n]&\rightarrow&\mathbb K(\mathfrak a)[x_2,\ldots, x_n]\\
&g(x_1,\ldots, x_n)&\mapsto&g(0,x_2,\ldots, x_n)\\
\\
\varphi_1:&\mathbb K(\mathfrak a)[x_1,\ldots, x_n]&\rightarrow&\mathbb K(\mathfrak a)[x_2,\ldots, x_n]\\
&g(x_1,\ldots, x_n)&\mapsto&g(1,x_2,\ldots, x_n)
\end{array}$$
Then $W(\sfI(\Fgenh, 1))$ can be identified with the disjoint union of the variety $V(\varphi_1(\sfI(\Fgenh, 1)))\subset \overline{\mathbb K(\mathfrak a)}^{n-1}$ and the projective variety $W(\varphi_0(\sfI(\Fgenh, 1)))\subset \mathbb P^{n-2}\overline{\mathbb K(\mathfrak a)}$.

\begin{itemize}
\item Notice that $\varphi_1(\sfI(\Fgenh,
  1))=\sfI(\varphi_1(\Fgenh),0).$ Therefore, the ideal $\varphi_1(\sfI(\Fgenh,
  1))\subset\overline{\mathbb K(\mathfrak a)}[x_2,\ldots,x_n]$ is
  spanned by $\varphi_1(\F^h)$ (which is a generic system of degree
  $D$ in $n-1$ variables) and by the maximal minors of its Jacobian matrix.  According to
  Lemma \ref{lem:smooth}, the variety $V(\varphi_1(\sfI(\Fgenh, 1)))$
  is empty.
\item Similarly, $\varphi_0(\sfI(\Fgenh, 1))=\sfI(\varphi_0(\Fgenh),0)\subset
  \mathbb K(\mathfrak a)[x_2,\ldots, x_n]$ is generated by the
  homogeneous polynomials $\varphi_0(\Fgenh)$ and by the
  maximal minors of the Jacobian matrix $\jac(\varphi_0(\Fgenh))$. Thus,  
  according to Corollary
  \ref{coro:smoothHom}, the variety $W(\varphi_0(\sfI(\Fgenh, 1)))$ is
  also empty.
\end{itemize}
\end{proof}

 \section{The homogeneous case}
\label{sec:hom}

In this section, our goal is to estimate the degree of regularity of
the ideal $\sfI(\Fgenh, 1)\subset \mathbb K(\mathfrak a)[X]$ which is
a homogeneous ideal generated by $\Fgenh$ and $\MaxM(\Fgenh, 1)$ (see
Notations \ref{notations:sec:prelim}). Recall that the degree of
regularity $\dreg(I)$ of a 0-dimensional homogeneous ideal $I$ is the
smallest positive integer such that all monomials of
degree $\dreg(I)$ are in $I$. Notice that $\dreg(I)$ is an upper bound
on the degrees of the polynomials in a minimal Gr\"obner basis of $I$
with respect to the grevlex ordering.

\begin{theo}\label{theo:hdreg}
The degree of regularity of the ideal $\sfI(\Fgenh, 1)$ is
$$\dreg(\sfI(\Fgenh, 1)) = D(p-1) + (D-2)n +2.$$
\end{theo}

\begin{notations}\label{not:2} To prove Theorem \ref{theo:hdreg}, we need to introduce a few more objects and notations. 
\begin{itemize}
  \item A set of new variables $\{u_{i,j} : 1\leq i\leq p, 2\leq j\leq n\}$ which is denoted by $U$;
  \item the determinantal ideal $\mathcal D\subset \K[U]$ generated by the maximal minors of the matrix
$$\left [\begin{array}{ccc}
  u_{1,2}&\dots&u_{1,n}\\
  \vdots&\vdots&\vdots\\
  u_{p,2}&\dots&u_{p,n}\\
\end{array}\right ].$$
\item $\frakg_1,\ldots, \frakg_{p(n-1)}\in\K(\mathfrak a)[U,X]$
  which denote the polynomials $u_{i,j}-\frac{\partial {\frakf}^h_i}{x_j}$,  for $1\leq i\leq
  p, 2\leq j\leq n$ and $\frakg_{p(n-1)+1},\ldots, \frakg_{pn}$ which denote the
  polynomials $\frakf^h_1,\ldots, \frakf^h_p$;
\item the ideals $\frakI_{(\ell)}=\mathcal D +\langle \frakg_1,\ldots,\frakg_\ell\rangle\subset \mathbb K(\mathfrak a)[U,X]$;
\item if $g\in\mathbb K[X]$ (resp. $I\subset\mathbb K[X]$) is a polynomial and $\prec$ is a monomial ordering (see e.g. \cite[Ch. 2, \S 2, Def. 1]{CoxLitShe97}), $\LM_\prec(g)$ (resp. $\LM_\prec(I)$) denotes its leading monomial (resp. the ideal generated by the leading monomials of the polynomials in $I$);
\item a \emph{degree ordering} is a monomial ordering $\prec$ such that for all pair of monomials $m_1,m_2\in \mathbb K[X]$, $\deg(m_1)<\deg(m_2)$ implies $m_1\prec m_2$.
\end{itemize}
\end{notations}
 Obviously the polynomials $\frakg_{k}$ for $1\leq k \leq p(n-1)$ will be used to mimic the process of substituting the new variables $u_{i,j}$ by $\frac{\partial {\frakf}^h_i}{x_j}$; indeed we have $\frakI_{(pn)}\cap \mathbb K[X]=\sfI(\Fgenh, 1)$.

Our strategy to prove Theorem \ref{theo:hdreg} will be to deduce the degree of regularity of $\sfI(\Fgenh,1)$ from an explicit form of its {\em Hilbert series}. 

Recall that, if $I$ is a homogeneous ideal of a polynomial ring $R$ with ground field $\K$, its \emph{Hilbert series} is the series
$$\HS_I(t)=\sum_{d\in\mathbb N} \dim_{\mathbb K}(R_d/I_d) t^d,$$
where $R_d$ denotes the $\mathbb K$-vector space of homogeneous polynomials of degree $d$ and $I_d$ denotes the $\mathbb K$-vector space $R_d\cap I$.

\begin{prop}\label{prop:HSunmixed}
The Hilbert series of the homogeneous ideal $\sfI(\Fgenh, 1)\subset\mathbb K(\mathfrak a)[X]$ is
$$\HS_{\sfI(\Fgenh, 1)}(t)=\frac{\det(A(t^{D-1}))}{t^{(D-1) \binom{p-1}{2}}}\frac{(1-t^D)^p (1-t^{D-1})^{n-p}}{(1-t)^n},$$
where $A(t)$ is the $(p-1)\times (p-1)$
matrix whose $(i,j)$-entry is $\sum_k \binom{p-i}{k} \binom{n-1-j}{k}
t^k$.
\end{prop}

The proof of Proposition \ref{prop:HSunmixed} is postponed to Section \ref{sec:HSproof}.

\begin{proof}[Proof of Theorem \ref{theo:hdreg}]
  By definition, the Hilbert series of a zero-dimensional homogeneous
  ideal is a polynomial of degree $\dreg -1$.  By Lemma
  \ref{lem:homzerodim}, $\sfI(\Fgenh, 1)$ has dimension $0$. Thus,
  using Proposition \ref{prop:HSunmixed}, we deduce that:
  { $$\dreg(\sfI(\Fgenh, 1)) =
    1+\deg\left(\frac{\det(A(t^{D-1}))}{t^{(D-1)
          \binom{p-1}{2}}}\frac{(1-t^D)^p
        (1-t^{D-1})^{n-p}}{(1-t)^n}\right).$$} The highest degree on
  each row of $A(t)$ is reached on the diagonal. Thus $\deg(\det
  A(t))=\frac{p (p-1)}{2}$ and a direct degree computation yields
$$\begin{array}{r@{\,}c@{\,}l}
  \dreg(\sfI(\Fgenh, 1))&=&1+\deg\left(\frac{\det(A(t^{D-1}))}{t^{(D-1) \binom{p-1}{2}}}\frac{(1-t^D)^p (1-t^{D-1})^{n-p}}{(1-t)^n}\right)\\
  &=&D(p-1) + (D-2)n +2.\end{array}$$
\end{proof}
{F}rom Proposition \ref{prop:HSunmixed}, one can also deduce the degree of $\sfI(\Fgenh, 1)$; this provides an alternate proof of \cite[Theorem 2.2]{NieRan09}. 
\begin{coro}\label{cor:deghom}
  The degree of the ideal $\sfI(\Fgenh, 1)$ is $$\DEG(\sfI(\Fgenh, 1))=\binom{n-1}{p-1} D^p  (D-1)^{n-p}.$$
\end{coro}
\begin{proof}
   By definition of the Hilbert series, the degree of the $0$-dimensional homogeneous ideal $\sfI(\Fgenh, 1)$ is equal to $\HS_{\sfI(\Fgenh, 1)}(1)$.
 By Proposition \ref{prop:HSunmixed}, direct computations show that $\HS_{\sfI(\Fgenh, 1)}(1)=\det(A(1)) D^p(D-1)^{n-p}$.  The determinant of the matrix $A(1)$ can be evaluated by using Vandermonde's identity and a formula due to Harris-Tu (see e.g. \cite[Example 14.4.14, Example A.9.4]{Ful97}).  We deduce that $\det (A(1))=\binom{n-1}{p-1}$ and hence $\HS_{\sfI(\Fgenh, 1)}(1)=\binom{n-1}{p-1} D^p (D-1)^{n-p}$.  \end{proof}

It remains to prove Proposition \ref{prop:HSunmixed}. This is done in the next subsections following several steps: \begin{itemize}
\item provide an explicit form of the Hilbert series of the ideal $\mathcal D$; this is actually already done in \cite{ConHer94}; we recall the statement of this result in Lemma \ref{lem:concaherzog};
  \item deduce from it an explicit form of Hilbert series of the ideal $\frakI_{(pn)}$ using genericity properties satisfied by the polynomials $\frakg_k$ and properties of quasi-homogeneous ideals; this is done respectively in Lemma \ref{lem:ndivzero} and Section \ref{subsec:quasihomogeneous};
  \item deduce from it the Hilbert series associated to $\sfI(\Fgenh, 1)$. 
\end{itemize}

\subsection{Auxiliary results}

\noindent
We start by restating a special case of \cite[Cor. 1]{ConHer94}.  
\begin{lem}[{\cite[Corollary 1]{ConHer94}}]\label{lem:concaherzog}  The Hilbert series of the ideal $\mathcal D\subset\mathbb K[U]$ is 
$$\HS_{\mathcal D}(t)=\frac{\det A(t)}{t^{ \binom{p-1}{2}}
  (1-t)^{n (p-1)}}.$$
\end{lem}

\begin{lem}\label{lem:ndivzero}
For each $2\leq\ell\leq n p$, $\frakg_\ell$ does not divide $0$ in $\mathbb K(\mathfrak a)[U,X]/\frakI_{(\ell-1)}$.
\end{lem}

\begin{proof}
  According to \cite[Thm. 2]{HocEag70}\cite{HocEag71}, the ring $\mathbb K(\mathfrak
  a)[U]/\mathcal D$ is a Cohen-Macaulay domain of Krull dimension
  $(n-1+p-(p-1))(p-1) = n (p-1)$.  Therefore, the ring $\mathbb
  K(\mathfrak a)[U,X]/\mathcal D$ is also a Cohen-Macaulay domain, and
  has dimension $n p$.

  Consider now the ideal $\langle \frakg_1,\ldots, \frakg_{n
    p}\rangle\subset(\mathbb K(\mathfrak a)[U]/\mathcal
  D)[X]$. According to Lemma \ref{lem:homzerodim}, the ideal
  $\sfI(\Fgenh, 1)=(\mathcal D + \langle \frakg_1,\ldots, \frakg_{n
    (p-1)}\rangle)\cap \mathbb K(\mathfrak a)[X]$ is
  zero-dimensional. Let $\prec$ denote a lexicographical monomial
  ordering such that for all $i,j,k$, $u_{i,j}\succ x_k$. Since the
  variables $U$ can be expressed as functions of $X$
  ($u_{i,j}-\frac{\partial f_i}{\partial x_j}\in \frakI_{(pn)}$), we have
  $\LM_\prec(\mathcal D + \langle \frakg_1,\ldots, \frakg_{n
    p}\rangle)=\langle u_{i,j}\rangle + \LM_\prec(\sfI(\Fgenh, 1))$ which is
  zero-dimensional. Therefore, the ideal $\mathcal D + \langle
  \frakg_1,\ldots, \frakg_{n p}\rangle\subset \mathbb K(\mathfrak
  a)[U,X]$ is zero-dimensional and hence so is $\langle
  \frakg_1,\ldots, \frakg_{n p}\rangle\subset\mathbb K(\mathfrak
  a)[U,X]/\mathcal D$.  Now suppose by contradiction that there exists
  $\ell$ such that $\frakg_\ell$ divides $0$ in $\mathbb K(\mathfrak a)[U,X]/\frakI_{(\ell-1)}$. Let
  $\ell_0$ be the smallest integer satisfying this property. Since
  $\mathcal D$ is equidimensional and $\forall \ell<\ell_0,
  \frakg_\ell$ does not divide $0$ in $\mathbb K(\mathfrak a)[U,X]/\frakI_{(\ell-1)}$, the ideal
  $\langle \frakg_1, \ldots, \frakg_{\ell_0-1}\rangle\subset
  \mathbb K(\mathfrak a)[U,X]/\mathcal D$ is equidimensional, has codimension $\ell_0-1$, and
  thus has no embedded components by the unmixedness Theorem
  \cite[Corollary 18.14]{Eis95}. Since $\frakg_{\ell_0}$ divides $0$
  in the ring $\mathbb K(\mathfrak a)[U,X]/(\mathcal D+\langle \frakg_1,\ldots,
  \frakg_{\ell_0-1}\rangle)$, the ideal $\langle \frakg_1, \ldots,
  \frakg_{\ell_0}\rangle\subset \mathbb K(\mathfrak a)[U,X]/\mathcal D$ has also codimension
  $\ell_0-1$. Therefore the codimension of $\langle \frakg_1, \ldots,
  \frakg_{n p}\rangle\subset \mathbb K(\mathfrak a)[U,X]/\mathcal D$ is strictly less than $n
  p$, which leads to a contradiction since we have proved that the
  dimension of this ideal is $0$. 
\end{proof}

\subsection{Quasi-homogeneous polynomials}\label{subsec:quasihomogeneous}  
The degrees in the matrix whose entries are the variables $u_{i,j}$ have to be balanced with $D-1$, the degree of the partial derivatives. This is done by changing the gradation by putting a \emph{weight} on the variables $u_{i,j}$, giving rise to \emph{quasi-homogeneous} polynomials. This approach has been used in \cite{FauSafSpa11a} in the context of the Generalized MinRank Problem. A polynomial $f\in \mathbb K[U,X]$ is said to be \emph{quasi-homogeneous} if the following condition is satisfied (see e.g. \cite[Definition 2.11, page 120]{GreLosShu07}):
$$f(\lambda^{D-1} u_{1,2},\ldots,\lambda^{D-1} u_{p,n},\lambda x_1,\ldots, \lambda x_k)=\lambda^d f(u_{1,2},\ldots,u_{p,n},x_1,\ldots, x_k).$$
The integer $d$ is called the weight degree of $f$ and denoted by $\wdeg(f)$.

An ideal $I\subset\mathbb K[U,X]$ is called
\emph{quasi-homogeneous} if there exists a set of quasi-homogeneous
generators of $I$. We let $\mathbb K[U,X]^{(w)}_d$
denote the $\mathbb K$-vector space of quasi-homogeneous
polynomials of weight degree $d$, and $I^{(w)}_d$ denote the set $\mathbb
K[U,X]^{(w)}_d\cap I$.  Ideals generated by quasi-homogeneous
polynomials are positively graded, as shown in \cite[Proposition 1]{FauSafSpa11a} that we restate below.

\begin{prop}[{\cite[Proposition 1]{FauSafSpa11a}}]
Let $I\subset\mathbb K[U,X]$ be an ideal.  Then the following statements are equivalent:
\begin{itemize}
\item there exists a set of quasi-homogeneous generators of $I$;
\item the sets $I^{(w)}_d$ are vector subspaces of $\mathbb K[U,X]^{(w)}_d$, and
$I=\bigoplus_{d\in \mathbb N} I^{(w)}_d$. 
\end{itemize}
\end{prop}

If $I$ is a quasi-homogeneous ideal, then $\mathbb K[U,X]/I$ is a graded algebra and hence its weighted Hilbert series
$\wHS_I(t)\in \mathbb Z[[t]]$ is well defined:
$$\wHS_I(t)=\sum_{d\in \mathbb N} \dim_{\mathbb K}(\mathbb K[U,X]^{(w)}_d/I^{(w)}_d) t^d.$$

The following lemma and its proof are similar to \cite[Lemma 5]{FauSafSpa11a}.

\begin{lem} \label{lem:HSeq}
  The Hilbert series of $\sfI(\Fgenh, 1)\subset\mathbb K(\mathfrak
  a)[X]$ is equal to the weighted Hilbert series of $\frakI_{(pn)}\subset \mathbb
  K(\mathfrak a)[X,U]$.
\end{lem}
\begin{proof}
  Let $\prec_{\mathsf{lex}}$ be a lex ordering on the variables of the polynomial ring $\mathbb
  K(\mathfrak a)[X,U]$ such that $x_k\prec_{\mathsf{lex}} u_{i,j}$ for all
  $k,i,j$. By \cite[Sec. 6.3, Prop. 9]{CoxLitShe97},
  $\HS_{\sfI(\Fgenh, 1)}(t)=\HS_{\LM_{\prec_{\mathsf{lex}}}(\sfI(\Fgenh, 1))}(t)$ and
  $\wHS_{\frakI_{(p(n-1))}}(t)=\wHS_{\LM_{\prec_{\mathsf{lex}}}(\frakI_{(p(n-1))})}(t)$.
  Since $\LM_{\prec_{\mathsf{lex}}}(u_{i,j}-f_{i,j})=u_{i,j}$ and $\frakI_{(pn)}\cap \mathbb K[X] = \sfI(\Fgenh, 1)$, we deduce that
  $$\begin{array}{rcl}\LM_{\prec_{\mathsf{lex}}}(\frakI_{(pn)})&=&\left\langle \{u_{i,j}\}\cup \LM_{\prec_{\mathsf{lex}}}(\frakI_{(pn)}\cap \mathbb K(\mathfrak a)[X])\right\rangle\\
    &=&\left\langle \{u_{i,j}\}\cup \LM_{\prec_{\mathsf{lex}}}(\sfI(\Fgenh, 1))\right\rangle.\end{array}$$
  Therefore, $\frac{\mathbb K(\mathfrak a)[U,X]}{\LM_{\prec_{\mathsf{lex}}}(\frakI_{(pn)})}$ is isomorphic (as a graded $\mathbb K(\mathfrak a)$-algebra) to $\frac{\mathbb K(\mathfrak a)[X]}{\LM_{\prec_{\mathsf{lex}}}(\sfI(\Fgenh, 1))}$. 

Thus, $\HS_{\LM_{\prec_{\mathsf{lex}}}(\sfI(\Fgenh, 1))}(t)=\wHS_{\LM_{\prec_{\mathsf{lex}}}(\frakI_{(pn)})}(t)$, and hence
  $\HS_{\sfI(\Fgenh, 1)}(t)=\wHS_{\frakI_{(pn)}}(t).$
\end{proof}

\subsection{Proof of Proposition \ref{prop:HSunmixed}}
\label{sec:HSproof}
We reuse Notations \ref{not:2}: $\sfI(\Fgenh, 1)= (\mathcal D +
\langle \frakg_1,\ldots, \frakg_{pn}\rangle)\cap \mathbb K(\mathfrak a)[X]$.  According to
Lemma \ref{lem:concaherzog} and by putting a weight $D-1$ on the variables $U$, the weighted Hilbert series of $\mathcal D\subset \mathbb
K(\mathfrak a)[U]$ is
$$\wHS_{\mathcal D\subset\mathbb
K(\mathfrak a)[U]}(t)=\frac{\det A(t^{D-1})}{t^{(D-1) \binom{p-1}{2}}
  (1-t^{D-1})^{n (p-1)}}.$$ 

Considering $\mathcal D$ as an ideal of
$\mathbb K(\mathfrak a)[X,U]$, we obtain
$$\wHS_{\mathcal D\subset\mathbb
K(\mathfrak a)[U,X]}(t)=\frac{1}{(1-t)^n}\wHS_{\mathcal D\subset\mathbb
K(\mathfrak a)[U]}(t).$$ 

If $I\subset \mathbb K(\mathfrak a)[U,X]$ is a quasi-homogeneous ideal
and if $g$ is a quasi-homogeneous polynomial of weight degree $d$
which does not divide $0$ in the quotient ring $\mathbb K(\mathfrak
a)[U,X]/I$, then the Hilbert series of the ideal $I+\langle g \rangle$
is equal to $(1-t^d)$ multiplied by the Hilbert series of $I$ (see e.g. the proof of \cite[Thm 1]{FauSafSpa11a} for more details).

Notice that the polynomials $\frakg_1,\ldots,\frakg_{p (n-1)}$ are quasi-homoge\-neous of weight degree $D-1$ (these
polynomials have the form $u_{i,j}-\frac{\partial \mathfrak f_i}{\partial x_j}$)
and the polynomials $\frakg_{p(n-1)+1},\ldots, \frakg_{pn}$ are quasi-homoge\-neous of weight degree $D$
(these polynomials are $\mathfrak f_1,\ldots, \mathfrak f_p$). Since $\frakg_\ell$ does not
divide $0$ in $\mathbb K(\mathfrak a)[U,X]/\frakI_{(\ell-1)}$ (Lemma \ref{lem:ndivzero}), the Hilbert series of
the ideal $\frakI_{(pn)}\subset
\mathbb K(\mathfrak a)[X,U]$ is
$$\begin{array}{r@{~}c@{~}l}
\wHS_{\frakI_{(pn)}}(t)&=&\displaystyle\frac{\det A(t^{D-1})(1-t^D)^p (1-t^{D-1})^{p (n-1)}}{t^{(D-1) \binom{p-1}{2}}
  (1-t^{D-1})^{n (p-1)} (1-t)^n} \\
&=&\displaystyle{\frac{\det A(t^{D-1})}{t^{(D-1) \binom{p-1}{2}}
 } \frac{(1-t^D)^p (1-t^{D-1})^{n-p}}{(1-t)^n}}.
\end{array}$$
Finally, by Lemma \ref{lem:HSeq}, $\wHS_{\frakI_{(pn)}}(t)=\HS_{\sfI(\Fgenh, 1)}(t)$.

 \section{The affine case}
\label{sec:affine}

The degree of regularity of a polynomial system is the highest degree
reached during the computation of a Gr\"obner basis with respect to
the grevlex ordering with the $F_5$ algorithm. Therefore, it is a
crucial indicator of the complexity of the Gr\"obner basis
computation.  On the other hand, the complexity of the FGLM algorithm
depends on the degree of the ideal $\sfI(\F,1)$ since this value is
equal to  $\dim_{\mathbb K}\left(\mathbb
  K[X]/\sfI(\F,1)\right)$.

In this section, we show that the bounds on the degree and the degree
of regularity of the ideal $\sfI(\Fgenh,1)$ are also valid for (not
necessarily homogeneous) polynomial families in $\mathbb K[X]$ under
genericity assumptions.

\begin{theo}\label{theo:dregaff}
  There exists a non-empty Zariski open subset $\mathscr O\subset
  \overline{\K}[X]^p_D$ such that, for any $\F$ in $\mathscr
  O\cap \K[X]^p$, 
$$\begin{array}{rcl}\dreg(\sfI(\F,1)) &\leq& D (p-1) + (D-2)n +2,\\
\DEG(\sfI(\F,1))&\leq&\binom{n-1}{p-1} D^p  (D-1)^{n-p}.\end{array}$$
\end{theo}

In the sequel, $\overline{\K}[X]_D$ denotes $\{f\in \overline{\K}[X]\mid \deg(f)=D\}$, and $\overline{\K}[X]_{D,\homg}$ denotes the homogeneous polynomials in $\overline{\K}[X]_D$.
In order to prove Theorem \ref{theo:dregaff} (the proof is postponed at the end of this section), we first need two technical lemmas.
\begin{lem}\label{lem:LMgener}
  There exists a non-empty Zariski open subset $\mathscr O\subset
  \overline{\K}[X]^p_{D,\homg}$ such that for all $\F^h\in \mathscr O\cap \K[X]^p$,
$\LM_\prec(\sfI(\F^h,1))=\LM_\prec(\sfI(\Fgenh,1))$.
\end{lem}

\begin{proof}
See e.g. \cite[Proof of Lemma 2]{FauSafSpa11a} for a similar proof.
\end{proof}

\begin{lem}\label{lem:deghomaff}
  Let $G=(g_1,\ldots, g_m)$ be a polynomial family and let
  $G^h=(g^h_1,\ldots, g^h_m)$ denote the family of homogeneous components of
  highest degree of $G$.  If the dimension of the ideal $\langle G^h \rangle$ is
  $0$, then $\DEG(\langle G\rangle)\leq \DEG(\langle G^h\rangle)$.
\end{lem}

\begin{proof}
  Let $\prec$ be an admissible degree monomial ordering. Let
  $\LM_\prec(h)$ denote the leading monomial of a polynomial $h$
  with respect to $\prec$.  Let $m\in
  \LM_\prec(\langle G^h\rangle)$ be a monomial. Then there exist
  polynomials $s_1,\ldots, s_m$ such that $\LM_\prec\left(\sum_{i=1}^m s_i g^h_i\right) = m.$
Since $\prec$ is a degree ordering,
$\LM_\prec\left(\sum_{i=1}^m s_i g_i\right) = m. $
Therefore $\LM_\prec(\langle G^h\rangle)\subset\LM_\prec(\langle
G\rangle)$.  If the ideal $\langle G^h \rangle$ is
$0$-dimen\-sional, then so is $\langle G\rangle$ and
hence $\DEG(\LM_\prec(\langle G\rangle))\leq\DEG(\LM_\prec(\langle
G\rangle))$. Since
$\DEG(I)=\DEG(\LM_\prec(I))$, we obtain
$\DEG(\langle G\rangle)\leq\DEG(\langle G^h\rangle).$\\
\end{proof}

\begin{proof}[Proof of Theorem {\ref{theo:dregaff}}]
  Let $\prec$ be a degree monomial ordering, and $\F^h=(f_1^h,\ldots,
  f_p^h)\in\overline{\K}[X]^p_{D,\homg}$ denote the homogeneous system
  where $f_i^h$ is the homogeneous component of highest degree of
  $f_i$.  By Lemma \ref{lem:LMgener}, there exists a non-empty Zariski
  subset $\mathscr O\subset \overline{\K}[X]^p_D$ such that, for any
  $\F$ in $\mathscr O\cap \K[X]^p$,
  $\LM_\prec(\sfI(\F^h,1))=\LM_\prec(\sfI(\Fgenh,1))$.  By \cite[Ch.9,
  \S 3, Prop.9]{CoxLitShe97}, the Hilbert series (and thus the
  dimension, the degree, and the degree of regularity) of a
  homogeneous ideal is the same as that of its leading monomial ideal.
  Hence, by Lemma \ref{lem:homzerodim},
$$\begin{array}{rcl}\dim(\sfI(\F^h,1))&=&\dim(\LM_\prec(\sfI(\F^h,1)))\\&=&\dim(\LM_\prec(\sfI(\Fgenh,1)))\\&=&\dim(\sfI(\Fgenh,1))=0.\end{array}$$
Similarly, by Theorem \ref{theo:hdreg}, $$\dreg(\sfI(\F^h,1))=\dreg(\sfI(\Fgenh,1))=D (p-1) +
  (D-2)n +2.$$

The highest degree reached during the $F_5$ Algorithm is upper bounded
by the degree of regularity of the ideal generated by the homogeneous
components of highest degree of the generators when this homogeneous
ideal has dimension 0 (see e.g. \cite{BarFauSalYan05} and references
therein).  Therefore, the highest degree reached during the
computation of a Gr\"obner basis of $\sfI(\F,1)$ with the $F_5$
Algorithm with respect to a degree ordering is upper bounded by
$$ \dreg\leq D (p-1) + (D-2)n +2. $$
\noindent
The bound on the degree is obtained by Corollary \ref{cor:deghom} and Lemma \ref{lem:deghomaff},
$$\begin{array}{rcl}\DEG(\sfI(\F,1))&\leq& \DEG(\sfI(\F^h,1))\\&\leq&\DEG(\LM_\prec(\sfI(\Fgenh,1)))\\&\leq&\binom{n-1}{p-1} D^p  (D-1)^{n-p}.\end{array}$$
\end{proof}

 \section{Complexity}
\label{sec:compl}
In the sequel, $\omega$ is a real number such that there exists an
algorithm which computes the row echelon form of $n\times n$ matrix in
$O(n^\omega)$ arithmetic operations (the best known value is
$\omega\approx 2.376$ by using Coppersmith-Winograd algorithm, see
\cite{Sto00}).

\begin{theo}\label{theo:compl}
There exists a non-empty Zariski open subset $\mathscr O\subset
  \overline{\K}[X]^p_D$, such that, for all $\F\in\mathscr O\cap \K[X]^p$,
the arithmetic complexity of computing a lexicographical Gr\"obner basis
of $\sfI(\F,1)$ is upper bounded by
 $$O\left(\binom{D (p-1)+(D-1) n+2}{D (p-1)+(D-2) n+2}^\omega + n
\binom{n-1}{p-1}^3 D^{3 p} (D-1)^{3 (n-p)} \right).$$
\end{theo}

\begin{proof}
  According to \cite{BarFauSal04,BarFauSalYan05}, the complexity of
  computing a Gr\"obner basis with the $F_5$ Algorithm with respect to
  the grevlex ordering of a zero-dimensional ideal is upper bounded by
$$O\left(\binom{n+\dreg}{\dreg}^\omega\right)$$
 where $\dreg$  
is the highest degree reached during the computation. In order to obtain a
lexicographical Gr\"obner basis, one can use the FGLM algorithm
\cite{FauGiaLazMor93}. Its complexity is $O\left(n
  \DEG(\sfI(\F,1))^3\right)$ (better complexity bounds are known in
specific cases, see \cite{FauMou11}).

According to Theorem \ref{theo:dregaff}, there exists a non-empty
Zariski open subset $\mathscr O\subset \overline{\K}[X]^p_D$ such
that, for all $\F$ in $\mathscr O\cap \K[X]^p$,
$$\begin{array}{rcl}\dreg(\sfI(\F,1)) &\leq& D (p-1) + (D-2)n +2,\\
  \DEG(\sfI(\F,1))&\leq&\binom{n-1}{p-1} D^p  (D-1)^{n-p}.\end{array}$$
Therefore, for all $\F$ in $\mathscr O\cap \K[X]^p$,
the total complexity of computing a lexicographical Gr\"obner basis
of $\sfI(\F,1)$: { $$O\left(\binom{D (p-1)+(D-1) n+2}{D
      (p-1)+(D-2) n+2}^\omega + n \binom{n-1}{p-1}^3 D^{3 p} (D-1)^{3
      (n-p)} \right).$$}
\end{proof}

\begin{coro}\label{coro:compln}
  If $D=2$, then there exists a
  non-empty Zariski open subset $\mathscr O\subset \overline{\K}[X]^p_2$, such that for all $\F\in\mathscr O\cap \K[X]^p$, the
  arithmetic complexity of computing a lexicographical Gr\"obner basis
  of $\sfI(\F,1)$ is upper bounded by 
$$O\left( \binom{n+2p}{2p}^\omega+n 2^{3p}\binom{n-1}{p-1}^3\right).$$
Moreover, if $p$ is constant and $D=2$, the arithmetic complexity is upper bounded by $O\left(n^{2p\omega}\right)$.
\end{coro}

\begin{proof}
This complexity is obtained by putting $D=2$ in the formula from
Theorem \ref{theo:compl}.
\end{proof}
In the sequel, the binary entropy function is denoted by $h_2$:
$$\forall x\in [0,1], h_2(x)=-x\log_2 (x)-(1-x)\log_2(1-x).$$
\begin{coro}\label{coro:complD}
  Let $D>2$ and $p\in\mathbb N$ be constant. There exists a non-empty
  Zariski open subset $\mathscr O\subset \overline{\K}[X]^p_D$, such that,
  for all $\F\in\mathscr O\cap \K[X]^p$, the arithmetic complexity of computing a
  lexicographical Gr\"obner basis of $\sfI(\F,1)$ is upper bounded
  by $$\displaystyle O\left(\frac{1}{\sqrt{n}}
    2^{(D-1)h_2\left(\frac{1}{D-1}\right)n \omega}\right)= O\left(
    (D-1)^{3.57 n}\right).$$
\end{coro}

\begin{proof}
Let $x$ be a real number in $[0,1]$. Then by applying Stirling's Formula, we obtain that
$$\binom{n}{x n}=O\left(\frac{1}{\sqrt{n}} 2^{h_2(x) n}\right).$$

Therefore, $$\begin{array}{rcl}\binom{(D-1)n}{n}&=&O\left(\frac{1}{\sqrt{n}}
    2^{(D-1)h_2\left(\frac{1}{D-1}\right) n}\right)\\
  &=& O\left(\frac{1}{\sqrt{n}} ((D-1) e)^{n}\right).\end{array}$$ Let
$C$ denote the constant $D (p-1)+2$. Then
$$\begin{array}{rcl}\binom{D (p-1)+(D-1) n+2}{D
    (p-1)+(D-2)
    n+2}&=&\binom{(D-1)n+C}{n}=O\left(\binom{(D-1)n}{n}\right)\\&=&O\left(\frac{1}{\sqrt{n}}
    2^{(D-1) h_2\left(\frac{1}{D-1}\right)n}\right).\end{array}$$ The
right summand in the complexity formula given in Theorem
\ref{theo:compl} is $O\left(n^{3p}(D-1)^{3n}\right)$ when $p$ and $D$
are constants; this is upper bounded by
$$O\left(\frac{1}{\sqrt{n}} 2^{(D-1)
    h_2\left(\frac{1}{D-1}\right)n\omega}\right).$$ Let
$\mathscr O$ be the non-empty Zariski open subset defined in Theorem
\ref{theo:compl}. For all $\F\in \mathscr O\cap \K[X]^p$, the arithmetic
complexity of computing a grevlex Gr\"obner basis of $\F$ is upper
bounded
by $$\begin{array}{r@{\,}c@{\,}l}O\left(\frac{1}{\sqrt{n}}
    2^{(D-1)h_2\left(\frac{1}{D-1}\right)n
      \omega}\right)&=&O\left(\frac{1}{\sqrt{n}}
    ((D-1)e)^{n\omega}\right)\\
  &=&O\left((D-1)^{\left(1+1/\log(D-1)\right)n\omega}\right)\\
  &=&O\left((D-1)^{3.57 n}\right),\end{array}$$ since $D\geq 3$ and
$\omega\leq 2.376$ with Coppersmith-Winograd algorithm.

On the other hand the asymptotic complexity of the FGLM part of the solving process is
$$O\left(n^{3(p-1)+1}(D-1)^{3n}\right)=\widetilde O\left(\left(D-1\right)^{3n}\right),$$
which is upper bounded by the complexity of the grevlex Gr\"obner basis computation.
\end{proof}

The following corollary shows that the arithmetic complexity is polynomial in the number of critical points.

\begin{coro}\label{coro:uniformbound}
  For $D\geq 3$, $p\geq 2$ and $n\geq 2$, There exists a non-empty Zariski open subset $\mathscr O\subset
  \overline{\K}[X]^p_D$, such that, for $\F\in\mathscr O\cap \K[X]^p$, the arithmetic complexity of
  computing a lexicographical Gr\"obner basis of $\sfI(\F,1)$ is
  upper bounded by
{$$O\left(\DEG\left(\sfI(\F,1)\right)^{{\max\left(\frac{\log(2eD)}{\log(D-1)}\omega,4\right)}}\right)\leq O\left(\DEG\left(\sfI(\F,1)\right)^{4.03 \omega}\right).$$}
\end{coro}

\begin{proof}
  Let $\mathscr O\subset \overline{\K}[X]^p_D$ be the non-empty
  Zariski open subset defined in Theorem \ref{theo:dregaff}, and
  $\F\in \mathscr O\cap \K[X]^p_D$ be a polynomial family.  First, notice that,
  since $p\geq 2$ and $n\geq 2$,
$$\begin{array}{rcl}
\DEG\left(\sfI(\F,1)\right)&=& \binom{n-1}{p-1} (D-1)^{n-p}D^p\\
&\geq& n
\end{array}$$
Therefore the complexity of the FGLM algorithm is upper boun\-ded by
$$O\left(n \DEG\left(\sfI(\F,1)\right)^3\right) \leq O\left(\DEG\left(\sfI(\F,1)\right)^{4}\right).$$
The complexity of computing a grevlex Gr\"obner basis of $\sfI(\F,1)$ is upper bounded by
$$\begin{array}{r@{~}c@{~}l}
\mathsf{GREVLEX}(p,n,D)&=& O\left(\binom{D(p-1)+(D-1)n+2}{n}^\omega\right)\\&\leq& O\left(\binom{2 D n}{n}^\omega\right).\end{array}$$
Notice that
$\binom{2 D n}{n}\leq (2 D)^n \frac{n^n}{n!}.$
By Stirling's formula, there exists $C_0$ such that
$\frac{n^n}{n!}\leq C_0 e^n$.  Hence $\mathsf{GREVLEX}(p,n,D)=
O\left((2 D e)^n\right)$.  

Since $D\geq 3$ and $n\leq \log(\DEG(\sfI(\F,1)))/\log(D-1)$, we obtain 
$$\begin{array}{rcl}O\left((2 D e)^{n\omega}\right)&\leq&
O\left(D^{\frac{\log(2 e D)}{\log D} n \omega}\right)\\&\leq&
O\left(\DEG\left(\sfI(\F,1)\right)^{\frac{\log(2 e D)}{\log
      (D-1)}\omega}\right).\end{array}$$
The function $D\mapsto \frac{\log(2 e D)}{\log (D-1)}$ is decreasing,
and hence its maximum is reached for $D=3$, and $\frac{\log(6 e)}{\log
  (2)}\leq 4.03$.
\end{proof}

Notice that in the complexity formula in Corollary
\ref{coro:uniformbound}, the exponent
$\frac{\log(2eD)}{\log(D-1)}\omega$ tends towards $\omega$ when $D$
grows. Therefore, when $D$ is large, the complexity of the grevlex
Gr\"obner basis computation is close to the cost of linear algebra
$O\left(\DEG(\sfI(\F,1))^\omega\right).$ Also, we would like to
point out that the bound in Corollary \ref{coro:uniformbound} is not
sharp since the formula $O\left(\binom{n+\dreg}{n}^\omega\right)$ for
the complexity of the $F_5$ algorithm is pessimistic, and the
majorations performed in the proof of Corollary
\ref{coro:uniformbound} are not tight.

 \section{Experimental Results}
\label{sec:experiments}
In this section, we report experimental results supporting the
theoretical complexity results in the previous sections.  Since
our complexity results concern the arithmetic complexity, we run
experiments where $\mathbb K$ is the finite field
$\mathsf{GF}(65521)$ (Figure \ref{fig:GFexp}), so that the timings represent the arithmetic
complexity. In that case, systems are chosen uniformly at random in $\mathsf{GF}(65521)[X]_D$.

We give experiments by using respectively the implementation of
$F_4$ and FGLM algorithms in the \textsc{Magma} Computer Algebra
Software, and by using the $F_5$ and FGLM implementations from the FGb
package.

All experiments were conducted on a 2.93 GHz Intel Xeon 
with 132 GB RAM.

\begin{figure}
\centering{
\begin{tabular}{|c|c|c|c|c|c|c|}
\hline
$n$&$p$&$D$&$\dreg$&$\DEG$&$F_4$ time&FGLM time\\
\hline
\hline
9&4&2&8&896&3.12s&18.5s\\
11&4&2&8&1920&61s&202s\\
13&4&2&8&3520&369s&1372s\\
15&4&2&8&5824&2280s&7027s\\
17&4&2&8&8960&10905s&$>$1d\\
30&2&2&4&116&3.00s&0.14s\\
35&2&2&4&136&7.5s&0.36s\\
40&2&2&4&156&13.3s&0.64s\\
6&4&3&17&3240&16s&400s\\
8&4&3&19&45360&35593s&$>$1d\\
7&2&3&12&1728&9.9s&91s\\
8&2&3&13&4032&121s&1169s\\
9&2&3&14&9216&736s&$>$1d\\
\hline
\end{tabular}
\caption{Experiments in \textsc{Magma} measuring the arithmetic complexity ($\mathbb K=\mathsf{GF}(65521)$).}\label{fig:GFexp}
}\end{figure}

\begin{figure}
\centering{ 
\begin{tabular}{|c|c|c|c|c|c|c|}
\hline
$n$&$p$&$D$&$\DEG(\sfI(\F,1))$&$F_5$ time&FGLM time&matrix density\\
\hline
\hline
15&3&2&728&1.38s&0.03s&36.86\%\\
16&3&2&840&2.20s&0.03s&36.91\%\\
17&3&2&960&3.21s&0.13s&36.96\%\\
18&3&2&1088&4.62s&0.12s&37.00\%\\
19&3&2&1224&6.57s&0.07s&37.04\%\\
20&3&2&1368&9.54s&0.10s&37.07\%\\
15&4&2&5824&131.65&10.66s&33.53\%\\
16&4&2&7280&258.6s&29.2s&33.78\%\\
17&4&2&8960&480.9s&68.9s&34.00\%\\
18&4&2&10880&877.36s&123.78s&34.19\%\\
19&4&2&13056&1600.1s&215.1s&34.35\%\\
20&4&2&15504&2727.6s&363.8s&34.49\%\\
21&4&2&18240&10371.7s&590.3s&34.62\%\\
9&1&3&768&0.32s&0.01s&22.45\%\\
10&1&3&1536&1.5s&0.15s&20.84\%\\
11&1&3&3072&8.5s&0.53s&20.59\%\\
12&1&3&6144&19.6s&2.46s&19.32\%\\
13&1&3&12288&276s&104s&19.12\%\\
14&1&3&24576&1759s&587s&18.08\%\\
7&2&3&1728&1.4s&0.14s&20.73\%\\
8&2&3&4032&13s&0.7s&20.26\%\\
9&2&3&9216&105s&37s&19.47\%\\
10&2&3&20736&909s&504s&19.08\%\\
6&3&3&2160&1.82s&0.12s&17.52\%\\
7&3&3&6480&31.3s&3.81s&17.39\%\\
6&4&3&3240&3.66s&0.49s&13.63\%\\
7&4&3&12960&140.2s&93.9s&14.55\%\\
8&4&3&45360&5126.9s&3833.9s&15.15\%\\
5&2&4&1728&0.84s&0.12s&14.46\%\\
6&2&4&6480&23.03s&2.01s&14.11\%\\
7&2&4&23328&634.0s&520.4s&13.64\%\\
8&2&4&81648&21362.6s&19349.4s&13.26\%\\
5&3&4&3456&3.58s&0.32s&11.36\%\\
6&3&4&17280&204.3s&139.7s&11.73\%\\
7&3&4&77760&13856.8s&16003s&11.83\%\\
\hline
\end{tabular}
\caption{Timings using the FGb library and $\mathbb K=\mathsf{GF}(65521)$.}\label{fig:fgb}
}\end{figure}

\textbf{Interpretation of the results.} Notice that the degree of
regularity and the degree match exactly the bounds given in
Theorem \ref{theo:dregaff}. In
Figures \ref{fig:GFexp} and \ref{fig:fgb}, we can see a different
behavior when $D=2$ or $D=3$. In the case $D=2 $, since the complexity
is polynomial in $n$ (Corollary \ref{coro:compln}), the computations
can be performed even when $n$ is large (close to $20$). Moreover, notice that for
$D=2$ or $D=3$, there is a strong correlation between the degree of
the ideal and the timings, showing that, in accordance with Corollary
\ref{coro:uniformbound}, this degree is a good indicator of the
complexity.

Also, in Figure \ref{fig:fgb}, we give the proportion of non-zero
entries in the multiplication matrices. This proportion plays an
important role in the complexity of FGLM, since recent versions of
FGLM take advantage of this sparsity \cite{FauMou11}. We can notice
that the sparsity of the multiplication matrices increases as $D$
grows.

\smallskip

\textbf{Numerical estimates of the complexity.} Corollary \ref{coro:uniformbound} states that the complexity of the grevlex Gr\"obner basis computation is upper bounded by $O\left(\DEG(\sfI(\F,1))^{4.03\omega}\right)$ when $D\geq 3$, $p\geq 2$, $n\geq 2$.
However, the value $4.03$ is not sharp. In Figure \ref{fig:numcompl}, we report numerical values of the ratio $\log\binom{n+\dreg}{n}/\log \left(\DEG(\sfI(\F,1))\right)$ which show the difference between $4.03$ and experimental values.

\begin{figure}
\centering
\begin{tabular}{|c|c|c|c|}
\hline
n&p&D&$\log\binom{n+\dreg}{n}/\log (\DEG)$\\
\hline
\hline
5&4&3&1.53\\
10&4&3&1.36\\
100&4&3&1.73\\
10000&4&3&1.99\\
10000&9999&3&2.28\\
30000&29999&3&2.28\\
1000&500&3&1.32\\
20000&2&3&2.00\\
500&250&1000&1.09\\
500&2&10000&1.11\\
\hline
\end{tabular}

\caption{Numerical values: $\log\binom{n+\dreg}{n}/\log\left(\DEG(\sfI(\F,1))\right)$.}\label{fig:numcompl}
\end{figure}

Notice that all ratios are smaller than $4.03$, as predicted by
Corollary \ref{coro:uniformbound}.  Experimentally, the ratio
decreases and tends towards 1 when $D$ grows, in accordance with the
complexity
formula $$O\left(\DEG\left(\sfI(\F,1)\right)^{\frac{\log(2eD)}{\log(D-1)}\omega}\right)$$
for the grevlex Gr\"obner basis computation. Also, when $D\geq 3$, the
worst ratio seems to be reached when $p=n-1$, $D=3$ and $n$ grows, and
experiments in Figure \ref{fig:numcompl} tend to show that it is
bounded from above by $2.28$.

\bibliographystyle{abbrv}
\bibliography{biblio}
\end{document}